\newif\ifcomments
\newif\ifanonymous
\newif\ifshort
\newif\iflncs

\iflncs
  \documentclass[runningheads]{llncs}
\else
  \documentclass[letterpaper,11pt,pdfa]{article}
  \usepackage[in]{fullpage}
\fi

\usepackage{iftex}
\ifPDFTeX
  \usepackage[utf8]{inputenc}
  \usepackage[noTeX]{mmap}
  \usepackage[T1]{fontenc}
\fi
\ifLuaTeX
  \usepackage{luatex85}
  \usepackage[noTeX]{mmap}
\fi

\usepackage{amsmath}
\usepackage{amsfonts}
\usepackage{amssymb}
\usepackage{color}
\usepackage{microtype}
\usepackage[bookmarks]{hyperref}
\usepackage[nameinlink]{cleveref}

\iflncs\else
  \usepackage{amsthm}
  \newtheorem{theorem}{Theorem}%
  \newtheorem{definition}[theorem]{Definition}
  
  \newtheorem{lemma}[theorem]{Lemma}
  \newtheorem{corollary}[theorem]{Corollary}
  \newtheorem{proposition}[theorem]{Proposition}
  
\fi

\usepackage[style=alphabetic,minalphanames=3,maxalphanames=4,maxnames=99,backref=true]{biblatex}
\renewcommand*{\labelalphaothers}{\textsuperscript{+}}
\renewcommand*{\multicitedelim}{\addcomma\space}
\DeclareFieldFormat{eprint:iacr}{Cryptology ePrint Archive: \href{https://ia.cr/#1}{\texttt{#1}}}
\DeclareFieldFormat{eprint:iacrarchive}{Cryptology ePrint Archive: \href{https://eprint.iacr.org/archive/#1}{\texttt{#1}}}
\AtEveryBibitem{%
 \clearlist{address}
 \clearfield{date}
 \clearfield{isbn}
 \clearfield{issn}
 \clearlist{location}
 \clearfield{month}
 \clearfield{series}

 \ifentrytype{book}{}{%
  \clearlist{publisher}
  \clearname{editor}
 }
}

\usepackage{enumitem}
\setlist[description]{noitemsep}
\setlist[enumerate]{noitemsep}
\setlist[itemize]{noitemsep}

\usepackage{soul}
\usepackage{xcolor}
\usepackage{xparse}
\makeatletter
\ExplSyntaxOn
\cs_new:Npn \white_text:n #1
{
  \fp_set:Nn \l_tmpa_fp {#1 * .01}
  \llap{\textcolor{white}{\the\SOUL@syllable}\hspace{\fp_to_decimal:N \l_tmpa_fp em}}
  \llap{\textcolor{white}{\the\SOUL@syllable}\hspace{-\fp_to_decimal:N \l_tmpa_fp em}}
}
\NewDocumentCommand{\whiten}{ m }
{
  \int_step_function:nnnN {1}{1}{#1} \white_text:n
}
\ExplSyntaxOff

\NewDocumentCommand{ \varul }{ D<>{5} O{0.2ex} O{0.1ex} +m } {%
  \begingroup
  \setul{#2}{#3}%
  \def\SOUL@uleverysyllable{%
    \setbox0=\hbox{\the\SOUL@syllable}%
    \ifdim\dp0>\z@
    \SOUL@ulunderline{\phantom{\the\SOUL@syllable}}%
    \whiten{#1}%
    \llap{%
      \the\SOUL@syllable
      \SOUL@setkern\SOUL@charkern
    }%
    \else
    \SOUL@ulunderline{%
      \the\SOUL@syllable
      \SOUL@setkern\SOUL@charkern
    }%
    \fi}%
  \ul{#4}%
  \endgroup
}
\makeatother

\newcommand{\E}{\mathop{\mathbb{E}}}

\usepackage{mleftright}
\newcommand{\paren}[1]{\left(#1\right)}

\newcommand{\mbracket}[1]{\mleft[#1\mright]}

\newcommand{\mbrace}[1]{\mleft\{#1\mright\}}
\newcommand{\abs}[1]{\left|#1\right|}

\usepackage{physics}
\newcommand{\UCRS}{\ket{\mathsf{UCRS}}}

\newcommand{\Commit}{\mathsf{Commit}}
\newcommand{\Reveal}{\mathsf{Reveal}}

\ifcomments
  \newcommand{\luowen}[1]{{\color{blue}Luowen: #1}}
\else
  \newcommand{\luowen}[1]{}
\fi

\begin{document}

\title{Unconditionally secure quantum commitments with preprocessing\thanks{This is the conference proceeding version. For a full-er version that contains the simplified proof of \Cref{thm:nonuniform}, please refer to the author's dissertation.}}
\iflncs\else
  \newcommand{\email}[1]{\href{mailto:#1}{\texttt{#1}}}
\fi
\ifshort
	\pagenumbering{gobble}
\fi
\iflncs
  \author{Luowen Qian\orcidID{0000-0002-1112-8822}}
  \institute{Boston University \& NTT Research \\\href{mailto:luowen@qcry.pt}{\email{luowen@qcry.pt}}}
\else
	\author{Luowen Qian\footnote{Boston University \& NTT Research. Email: \email{luowen@qcry.pt}}}
\fi
\date{}
\maketitle

\begin{abstract}
	We demonstrate how to build computationally secure commitment schemes with the aid of quantum auxiliary inputs \emph{without unproven complexity assumptions}.
	Furthermore, the quantum auxiliary input can be either sampled in uniform exponential time or prepared in at most doubly exponential time, without relying on an external trusted third party.
	Classically, this remains impossible without first proving $\mathsf{P} \neq \mathsf{NP}$.
\end{abstract}

\ifshort
	\newpage
	\pagenumbering{arabic}
\fi

\section{Introduction}

It has been known since 1990 \cite{IL89-essential,Gol90-ci} that almost all interesting classical cryptographic tasks requires computational security, and furthermore, hardness assumptions that are at least as strong as the existence of one-way functions.
Thus realizing these cryptographic tasks unconditionally faces the barrier of ``$\mathsf P \stackrel{?}{=} \mathsf{NP}$'', a problem that has undergone intense studies by complexity theorists.
These cryptographic tasks in particular include constructing a commitment scheme, the feasibility of which is equivalent to the existence of one-way functions.

Auxiliary-input cryptography, studied since the 1990s \cite{OW93-zk}, is a non-uniform version of cryptography where every party in the protocol gets access to a copy of some public information that might not be efficiently preparable.
This is not to be confused with non-uniform security, which is the default security notion where the adversaries, in addition to running in polynomial time, get some advice at the beginning from an inefficient preprocessing phase or perhaps some residual information from another protocol execution.
Following the same proofs, the same barrier of ``$\mathsf P \stackrel{?}{=} \mathsf{NP}$'' still applies in this more relaxed setting.

Given this difficulty, it is natural to consider constructing \emph{quantum} commitments instead.
Recent works have demonstrated that quantum commitments play a similar central role as classical ones, in terms of its tight connections to both quantum cryptography at large \cite{Yan22-commitment,BCQ23-efi,BEMPQY23-uhlmann} and quantum complexity \cite{BEMPQY23-uhlmann}.
While commitments statistically (or information theoretically) secure against both parties are impossible even quantumly \cite{May97-commitment,LC97-commitment}, recent works have demonstrated that computationally secure ones are possible under complexity assumptions \cite{BCQ23-efi,BEMPQY23-uhlmann,Bra23-black} that are evidently milder than $\mathsf{P} \neq \mathsf{NP}$ \cite{Kre21-pseudorandom,AQY22-prs,MY22-commitment,KQST22-algorithmica,LMW23-unitarysynth}.
This line of works suggests that achieving computationally secure quantum cryptography might not be susceptible to the same barriers that apply to classical cryptography.

Nonetheless, it is still reasonable to speculate that any reasonable quantum computational cryptography could face some other barriers.
Indeed, all prior quantum computational cryptography still starts by assuming some hardness assumptions, even though they may be weaker than what is needed classically as was shown above.

In this work, we consider a natural quantum non-uniform cryptography notion called quantum auxiliary input, meaning that every party receives copies of the same quantum pure state as input.%
\footnote{
  There are two different variants of auxiliary-input security considered in the literature: in this work, we focus on the strong variant where the adversary's success probability is small for all but finitely many auxiliary inputs; however, there is also a weaker variant where we only require the adversary's success probability to be small for infinitely many auxiliary inputs \cite{OW93-zk}.
  Note that classical auxiliary input \emph{weakly} secure quantum commitments can be built assuming $\mathsf{QCZK} \not\subseteq \mathsf{BQP}$ \cite{BCQ23-efi}, classical auxiliary-input (strong) quantum commitment can be built assuming $\mathsf{QCZK} \not\subseteq \mathsf{QMA}$, and finally standard quantum commitments without auxiliary input can be built assuming various other unproven assumptions.
}
Therefore, in some sense, this quantum auxiliary information can be thought of as a ``cryptographic magic state'', extremely similar to magic states that occur in quantum fault tolerance, where a piece of quantum state can augment the computational power of a less powerful circuit family.

Our main theorem, perhaps surprisingly, constructs quantum auxiliary-input commitments without unproven complexity assumptions.
\begin{theorem}
    \label{thm:main}
    There exists a computationally-hiding statistically-binding non-interactive quantum commitment scheme with quantum auxiliary input.
    Furthermore, the quantum auxiliary input has an exponential-size classical description that can be sampled uniformly in exponential time.
\end{theorem}
Our proof builds upon the prior work of Chailloux, Kerenidis, and Rosgen~\cite{CKR16-commitments} who established the same result assuming the unproven complexity separation that $\mathsf{QIP} \not\subseteq \mathsf{QMA}$.
However, our construction uses a sparse pseudorandom ensemble constructed with a probabilistic method instead of the assumed hardness of a $\mathsf{QIP}$-complete problem.

Despite their similarities, the fundamental nature of this theorem is different since it is unconditional.
As a result, there are a few interpretations of this theorem.
\begin{enumerate}
  \item This is the \emph{first} demonstration of a \varul{useful cryptographic task with unconditional inherently-computational security} (quantum or classical), meaning that such tasks are impossible with statistical security.
  Cryptographers are often trained to instinctively assume that computational security protocols rely on hardness assumptions, and conversely unconditionally secure protocols are statistically secure.
  This work, however, reveals that such presumptions are unwarranted, especially in the quantum setting.
  \item \Cref{thm:main} also reveals that we should be more cautious about the speculation that quantum computational cryptography still requires making hardness assumptions, and perhaps even reassess its validity and \varul{investigate the existence of barriers}.
  \item This is also arguably a conditional demonstration of \varul{quantum computational advantage through cryptography}.
  As discussed, a classical analogue of \Cref{thm:main} would still imply $\mathsf{P} \neq \mathsf{NP}$ \cite{IL89-essential}.
  Conversely, this shows that quantum auxiliary-input commitments are feasible even if $\mathsf{P} = \mathsf{NP}$ and classical (auxiliary-input) commitments are impossible.
  This is different from traditional conditional quantum computational advantages (such as factoring) where the condition is classical easiness rather than classical hardness of certain problems.
\end{enumerate}

In fact, it is possible to strengthen this classical impossibility above further.
As we show in \Cref{thm:classical-impossibility}, this is true even if we consider the model where all parties get access to a single (possibly inefficient) sampling oracle, and each sample is private to the requested party.
Here we allow access to randomized auxiliary information in order to level the playing field a bit for classical protocols, since a pure quantum state can also implement such randomized samples by taking an appropriate superposition over bitstrings.
This model is even stronger than having access to classical auxiliary input, since the oracle could just output that fixed string with probability 1.
(Considering every party having access to the same randomized advice can be simply replaced with a fixed classical advice by an averaging argument.)
Thus in some sense, our result is even stronger than Raz's result of $\mathsf{QIP/qpoly} = \mathsf{IP/rpoly} = \mathsf{ALL}$ \cite{Raz05-qpcp} as the power of the quantum advice does not come from the advice being inherently randomized and this randomness being private to each party.

\paragraph{Removing trust.}
One might be skeptical of the security since na\"ively it appears that the parties need to assume they can trust the quantum auxiliary input given to them.
One way to address this concern is to ask the skeptical party to simply inspect the classical description of the magic state and verify that the commitment built with it is secure.
This takes at most a doubly exponential time since we can cast it as a $\mathsf{QMA}$ problem (\Cref{prop:qma-security}).
Alternatively, we could also ask them to find the lexicographically first magic state that works, which also takes at most doubly exponential time.

In \Cref{sec:trustless}, we show how to achieve computationally secure commitments without any trusted auxiliary inputs through a more efficient preprocessing.
Specifically, in the scheme we only need to ask both the committer and the receiver to perform a uniform exponential-time preprocessing phase.

\paragraph{Removing inefficiency.}
Complementary to that, we also show how to do commitments in a \emph{completely efficient} setting with a weakly trusted setup in \Cref{sec:crs}.
Specifically, we adapt our construction into a trusted setup model where a trusted third party efficiently generates a few copies of the quantum auxiliary information for every party in the protocol to use before the protocol begins.
However, unlike \Cref{thm:main}, here the pure state distributed to every party is not deterministic.
Furthermore, the scheme could in fact be statistically secure against all parties if the number of copies distributed is restricted.

In comparison, a stronger (in terms of trust) setup model is the secret parameter model \cite{Ps05-nizk}, where we need to trust the setup to sample two correlated secret strings for each party.
While statistically secure classical non-interactive commitments are possible in the secret parameter model, our model appears to be meaningfully weaker than the secret parameter model.
In particular, we can again invoke our classical impossibility (\Cref{thm:classical-impossibility}) to argue that statistically secure commitments remain impossible in a classical sampling analogue of our model, since an unbounded adversary could always solve any $\mathsf{NP}$ problem.

\paragraph{Polynomially bounded adversaries are physical, probably.}
A possible concern regarding the claimed ``unconditional'' nature of this result is that the security relies on the ``assumption'' that the adversary is polynomially bounded during the execution of the protocol.
We address that concern by noting that it is possible to reduce this assumption to physical assumptions.
On one hand, physical assumptions and more generally modeling assumptions (the mapping between real world and mathematics) are unavoidable in any form of provably secure cryptography.
On the other hand, this suggests a new win-win philosophy, since either we can have secure cryptography or we are able to discover exciting new physics.

More specifically, there is most likely a fundamental physical limit to the density of quantum information before collapsing into a black hole given by the Bekenstein bound (see \cite{Llo00-limits} for an exposition of this).
Since any malicious computation must be done in polynomial time, and thus space by the no superluminal signaling principle, it follows from extended Church--Turing thesis (for quantum information tasks) that any polynomial-time physical computation can be described by a polynomial-size quantum circuit.
Therefore, we can force adversary to be polynomially bounded by simply limiting the protocol execution time.
Even if the adversary could somehow leverage black holes to perform useful computations, we could still monitor the energy density nearby to make sure that this does not happen.
We leave further materializing this idea to future work.

\paragraph{Additional applications.}
We note that it is possible to use this commitment scheme to further instantiate oblivious transfer (OT) and secure multiparty computations (MPC), as we further detail in \Cref{sec:simulation} (see also \cite{BCQ23-efi}).
This is in contrast with Kent's statistically secure relativistic commitment scheme \cite{Kent99-commitment,Kent12-commitment} since (statistical) OT and MPC are known to be impossible even in that quantum relativistic model \cite{Rudolph02-laws,Colbeck06-mpc}.
(Another comparison is that our model also does not impose any strict constraints on the physical location of the parties: they simply need to be a polynomial distance away from each other so that the polynomial timing constraints can be satisfied.)

For a more concrete example for MPC, consider the classical Yao's Millionaires' Problem \cite{Yao82-mpc}, except that now it's a Trillionaires' Problem!
This means that two trillionaires want to figure out who is richer without revealing anything else.
Also since they are trillionaires, their entire families and businesses are somehow also on the line, so if cheating is detected then chaos would ensue.
Furthermore, they also have the world's best cryptanalyst to break any computational hardness assumption should it be necessary and possible, and they are willing to use exponential-time preprocessing as a small sacrifice.

Secure multiparty computations with preprocessing is \emph{the perfect} solution for this problem!
They first spend exponential time to set up and during the protocol execution, they ensure that the other party finish in time and no cheating occurs.
So if nothing bad occurs, then both of them can be satisfied knowing that their secrets are safe.
Additionally, using certified everlasting transfer \cite{BK23-certifieddel}, they can achieve everlasting security by having a third party referee, who is trusted to be uninterested in spending exponential resources recovering the input, to certifiably delete the remaining information.

\paragraph{Open problems.}

\ifshort
By scaling everything down exponentially, we can also view our scheme as a sublinear cryptographic scheme.
In particular, we give a commitment scheme of preprocessing phase of time $T$ and execution phase of time $\Theta(\log T)$, with security $\tilde O(\sqrt[3]{S/T})$ against any online adversary of size $S$.
Can we have a provably secure commitment scheme with a better preprocessing-time-vs-security trade-off?
For example, if we shave off the cube root from \Cref{thm:nonuniform} then the security can be improved to $\tilde O(S/T)$.
\fi

One undesirable feature of our scheme is that there is no good way to get additional copies of the magic state.
In the trusted setup model, it is possible to redo the setup every once in a while to preserve statistical security and efficiency, but otherwise generating new copies still takes exponential time, although it is possible to build dedicated hardware to pipeline this process.
A natural question is whether we can construct quantum cryptography with quantum auxiliary information that is efficiently clonable or reusable.
While there are families of quantum states that is efficiently clonable but cannot be prepared in uniform polynomial time relative to a quantum oracle \cite{BZ23-teleportation}, even having a standard model candidate is open since the quantum oracle there is a cloning oracle.
The biggest issue with reusing our schemes is that for each invocation of the commitment scheme, on average $\Theta(\lambda)$ copies of the magic state is transferred to the other party, and there does not appear to be a way to certifiably retrieve these states.
Of course, the ultimate goal would be to construct these without inefficient preprocessing or having trusted setup at all.

\paragraph{Concurrent work.}
\ifanonymous
\else
Near the completion of this work, I became aware of a parallel work done by Tomoyuki Morimae, Barak Nehoran, and Takashi Yamakawa~\cite{MNY23-qaicomm}.
In particular, both works have the same constructions of quantum auxiliary-input commitments (\Cref{thm:main}) and statistically secure commitments with a weakly trusted (stateful) setup (\Cref{cor:ucrs-comm} and \Cref{prop:stathiding}).
They have an additional statistical commitment construction with stateless setup, however, the downside is that this can only be secure given that at most a bounded number of copies is generated.
Complementary to the classical impossibility of \Cref{thm:classical-impossibility}, they also give an impossibility of classical auxiliary-input (quantum) commitments in a weak setup model.
They in addition have the application of quantum auxiliary-input zero-knowledge proofs for $\mathsf{NP}$ (with negligible simulation security).
They also point out an observation by Fermi Ma that the cube root security loss of \Cref{thm:nonuniform} can be improved to square root if we instead augment the \cite{GK90-pseudorandom} argument with a matrix Chernoff bound like was done in \cite{LMW23-unitarysynth}.
This square root loss is also tight even against classical algorithms \cite[page 651]{DTT10-attacks}.

Notably, their work point out that the quantum auxiliary-input commitment cannot be immediately used to construct simulation-secure commitments through the \cite{BCKM21-mpc} compiler due to the use of Watrous rewinding in the simulator there.
After the discussions with \cite{MNY23-qaicomm}, I present the resolution of this in \Cref{sec:simulation} by showing how to adapt \cite{BCKM21-mpc} to recover $\varepsilon$-simulation security, which still suffices to recover almost all applications of a standard quantum commitment.
\fi

\ifshort
\section{Non-interactive quantum commitments}
\label{sec:commitment-def}

In this section, we use $x \in S$ as the auxiliary information for a certain set $S$ and implicitly its length $|x|$ as  the security parameter.
Thus if $x$ is a classical unary string, then it is a standard uniform commitment scheme.
In order to prevent degeneracy, we require that $S$ must contain arbitrarily long bitstrings or quantum states: for any integer $n$, there exists $x \in S$ such that $|x| \ge n$.
We refer the readers to \cite{BCQ23-efi} for quantum information and cryptography background.

\begin{definition}
	A non-interactive commitment scheme is a pair of efficient quantum algorithms $\Commit$ and $\Reveal$ where $\Commit(x, b)$ produces a bipartite state $\rho$ over the commitment register $\mathsf C$ and the decommitment register $\mathsf D$, and $\Reveal(x, \rho)$ outputs either the committed bit $b'$ or a rejection symbol $\bot$.
	Furthermore, $\Pr[\Reveal(x, \Commit(x, b)) \neq b]$ is negligible for any $b = 0, 1$ and $x \in S$.

	We say it is computationally (or statistically) hiding if the $\mathsf C$ registers of $\Commit(x, 0)$ and $\Commit(x, 1)$ are computationally (or statistically, respectively) indistinguishable.

	We say it is statistically (or computationally) sum binding if for any state $\rho$ (over $\mathsf C, \mathsf D$ and a private register $\mathsf M$) and any possibly inefficient (or efficient, respectively) unitary $U$ not acting on register $\mathsf C$, $p_0 + p_1 - 1$ is negligible, where $p_b$ is the probability that the receiver accepts bit $b$.
	In particular, $p_0 := \Pr[\Reveal(x, \rho_{\mathsf{CD}}) = 0]$ and $p_1 := \Pr[\Reveal(x, (U\rho)_{\mathsf{CD}}) = 1]$.

	Finally, we say it is in canonical form \cite{Yan22-commitment}, if $\Reveal$ takes the following form:
	\begin{enumerate}
		\item Perform a rank-1 projection on $\mathsf{CD}$ and output $0$ if it succeeds.
		\item Perform another orthogonal rank-1 projection on $\mathsf{CD}$ and output $1$ if it succeeds.
		\item Output $\bot$.
	\end{enumerate}
\end{definition}

Note that here $\rho$ can be arbitrarily inefficient even in the case of computational hiding, which captures the standard non-uniform security notion where the adversary can have potentially a different auxiliary input for every security parameter.

Since statistical sum binding is traditionally unwieldy to use in applications, the work of Ananth, Qian, and Yuen defines an extractor-based binding definition, which on a high level states that for any (non-uniform) malicious committer, there is an extractor that can extract the committed bit from the receiver's view after the commitment phase in an imperceptible way, as long as the receiver does not touch those registers until the reveal phase.

\newcommand{\realexpt}{\mathsf{RealExpt}}
\newcommand{\idealexpt}{\mathsf{IdealExpt}}
\begin{definition}[{\cite[Definition 6]{AQY22-prs}}]
	\label{def:stat:binding}
	We say that a quantum commitment scheme satisfies \emph{receiver-extractable statistical binding} if for any malicious (non-uniform) committer $\mathcal C$, there exists a (possibly inefficient) extractor algorithm ${\cal E}$ such that $\realexpt$ is statistically indistinguishable from $\idealexpt$, where
	\begin{itemize}
		\item $\realexpt$: Execute the commit phase to obtain the joint state $\sigma_{C R}$. Execute the reveal phase to obtain the trit $\mu$ outputted by the receiver. Output the pair $(\tau_{C},\mu)$ where $\tau_{C}$ is the final state of the committer.
		\item $\idealexpt$: Execute the commit phase to obtain the joint state $\sigma_{C R}$.  Apply the extractor $I_C \otimes {\cal E}_R(\sigma)$ to obtain a new joint committer-receiver state $\sigma'_{CR}$ along with the extracted trit $b' \in \{0,1,\bot\}$.
		Starting from $\sigma'$, execute the reveal phase to obtain the trit $\mu$. Let $\tau_{C}$ denote the final state of the committer.
		If $\mu = \bot$ or $\mu = b'$, then output $(\tau_{C},\mu)$.
		Otherwise, output a special symbol $\mathfrak{E}$ (unused in the real experiment) indicating extraction error.
	\end{itemize}
	\end{definition}

For canonical form commitments, it is known that many variants of binding are equivalent, including statistical sum binding and extractable binding \cite{FUYZ22-binding}.
It is unclear how to make commitments constructed in this work into canonical form due to the presence of quantum auxiliary information, so for completeness we show how to extend the equivalence to general non-interactive schemes.

We call a commitment scheme having a projective $\Reveal$ if $\Reveal$ is a projective measurement on the receiver's view.
For example, any post-quantum commitment with a deterministic $\Reveal$ algorithm or any canonical-form commitment has a projective $\Reveal$ but our schemes do not due to the swap test.
Without loss of generality we can always generically make a commitment have a projective $\Reveal$ via Stinespring dilation.
In particular, we need to purify the $\Reveal$ algorithm so that it is a unitary on registers $\mathsf{CD}$ and potentially some auxiliary register $\mathsf{A}$ (which holds $x$ and potentially some zeroes), followed by a complete measurement on a qutrit in the auxiliary register to obtain the output.
Furthermore, we ask the receiver to prepare the auxiliary register $\mathsf A$ immediately after the commit phase.
This change is also imperceptible from the committer's perspective.

We remark that due to a technicality, the equivalence cannot hold without giving the extractor access to $\mathsf A$.
Consider the following unnatural counterexample, a non-interactive (classical) commitment scheme without a projective $\Reveal$ that satisfies sum binding but not extractable binding.
To commit to bit $b$, the committer simply send a mode bit $1$ followed by $b$; and the decommitment message is empty.
The $\Reveal$ algorithm checks that if the mode bit is $1$ then output $b$, and if the mode bit is $0$ then output a random bit.
Intuitively this scheme is binding since no matter which mode the malicious committer uses, he cannot later change the bit in any way since the decommitment message is empty; and indeed it is straightforward to see that it satisfies sum binding.
However, a malicious committer can cause the extractor to fail since without access to the random bit register later used to sample the output of $\Reveal$, it is impossible for the extractor to predict the bit that would be revealed later.

\begin{theorem}
	Statistical sum binding is equivalent to extractable binding for any commitment scheme with a projective $\Reveal$, thus we can build an extractable binding commitment scheme from any statistical sum binding commitment scheme (and vice versa).
\end{theorem}
\begin{proof}
	Extractable binding implies sum binding is straightforward: $p_0 + p_1 \le 1$ in the ideal experiment where the bit is extracted and guaranteed to be correct, thus the indistinguishability between the real and the ideal experiments implies sum binding.

	For the other direction, fix any sender where the overall state before executing $\Reveal$ is $\rho$.
	Without loss of generality we assume $\rho$ is pure by taking the purification into committer's private register.
	Let $\ket{\rho_0}$ be the state post-selected on the event that $\Reveal(\ket{\rho}) = 0$ and $\ket{\rho_1}$ be the state post-selected on the event that $\Reveal(\ket{\rho}) = 1$, or $0$ if such post-selection is not possible.
	$\ket{\rho_0}$ and $\ket{\rho_1}$ are orthogonal by the fact that $\Reveal$ is projective.
	Consider a canonical-form variant of $\Reveal$ where the rank-1 projections are given by $\ketbra{\rho_0}{\rho_0}$ and $\ketbra{\rho_1}{\rho_1}$.
	We claim that this variant is still sum binding since any sum binding adversary succeeding for this variant would also succeed in breaking sum binding against the original $\Reveal$.
	Then the result of Fang, Unruh, Yan, and Zhou~\cite{FUYZ22-binding} (also in \cite[Appendix B]{MY22-commitment}) implies that there is an extractor for this modified scheme, in particular, the extractor performs the optimal distinguishing measurement between these two states on registers $\mathsf{AC}$.
	This extractor also works for the original scheme since the extractor almost perfectly project the state onto either $\ket{\rho_0}$ or $\ket{\rho_1}$ by statistical sum binding and gentle measurement, and the two $\Reveal$ algorithms behave identically in the subspace spanned by these states.
\end{proof}

\fi

\section{Quantum auxiliary-input commitment}

\ifshort\else
We formally define non-interactive commitments and handle its subtleties for later applications in \Cref{sec:commitment-def}.
\fi
To prove our main theorem, we first need the following result on the quantum non-uniform security of pseudorandomness.

\begin{proposition}[{\cite{CGLQ20-tradeoffs,Liu23-advice}}]
	\label{thm:nonuniform}
	For a random function $H: [N] \to [M]$, the best quantum circuit of size $S$ (potentially depending on $H$) can distinguish its output from a random element from $[M]$ with advantage at most $12 \cdot \sqrt[3]{\frac SN}$ (averaged over the choice of $H$).
\end{proposition}

While these works consider the more general case where quantum algorithms that could additionally make queries to $H$, the same proofs also give a polynomial upper bound for standard query-less algorithms.
For completeness, we show this precise bound in \Cref{sec:nonuniform-proof}.

In particular, this implies the existence of a good function for which this is true against size $S - 1$, since we can always use an extra bit to make the bias always have the same sign for all $H$.
In other words, for any $(S - 1)$-size algorithm $A'$, there is an $S$-size algorithm $A$ such that
\[ \abs{\E_H\mbracket{\E_x[A'(H(x))] - \E_y[A'(y)]}} = \E_H\mbracket{\abs{\E_x[A(H(x))] - \E_y[A(y)]}}, \]
where $A$ simply runs $A'$ and XOR its output with the extra advice bit.
With some small inverse exponential security loss, we can also use Markov's inequality to argue that even sampling a function uniformly at random would satisfy this with overwhelming probability.

Instantiating the above with $S = 2^\lambda + 1, N = 2^{5\lambda}$ and $M = 2^{6\lambda}$, we get the following corollary, which generalizes the classical (inefficient) sparse pseudorandom ensemble construction of Goldreich and Krawczyk~\cite{GK90-pseudorandom} to the post-quantum setting\footnote{The na\"ive attempt of quantizing \cite{GK90-pseudorandom} proof cannot seem to handle quantum advice. For readers familiar with that work, a natural strategy is to use union bound over all quantum algorithms/advice using an $\varepsilon$-net, however, this fails since the size of the $\varepsilon$-net is doubly exponential, and thus eliminating the single exponential concentration we get from Hoeffding's bound.}.

\begin{corollary}[Exponentially secure sparse pseudorandom ensemble]
	There exists a pseudorandom ensemble $H: \{0, 1\}^{5\lambda} \to \{0, 1\}^{6\lambda}$ against all $2^\lambda$-size quantum circuits $A$ with quantum auxiliary information, whose security $\abs{\E_x[A(H(x))] - \E_y[A(y)]} \le 2^{-\lambda}$ for all $\lambda \ge 11$.
\end{corollary}

\begin{proof}[{Proof of \Cref{thm:main}}]
	Let $\{H(x)\}_{x \in \{0, 1\}^{5\lambda}}$ be the exponentially secure sparse pseudorandom ensemble as above.
	We first specify the quantum auxiliary input
	\[
	\ket{M_\lambda} := 2^{-5\lambda/2} \sum_{x \in \{0, 1\}^{5\lambda}} \ket{x} \otimes \ket{H(x)},
	\]
	a pure state of $11\lambda$ qubits that can be prepared using a circuit of size $\tilde O(2^{5\lambda})$.

	We now specify the protocol.
	The sender, to commit to 0, simply sends the second half of $\ket{M_\lambda}$; and to commit to 1, sends the second half of a maximally entangled state
	\[\ket{\Psi_\lambda} := 2^{-3\lambda} \sum_{y \in \{0, 1\}^{6\lambda}} \ket y \otimes \ket y;\]
	to later decommit, the rest of the state is sent.
	The receiver, upon receiving the entire state and the bit $b$ can efficiently test by performing a SWAP test between the received state and the correct state of either $\ket{M_\lambda}$ or $\ket{\Psi_\lambda}$.
	Finally, a $\lambda$-fold parallel repetition is applied to this construction, meaning that the committer commits to the same bit $\lambda$ times in parallel, and the receiver checks that all commitment decommits to the same bit.

	For binding, we note that our construction is identical to that of \cite[Section 4]{CKR16-commitments} except for the choice of two pure states used for the commitments.
	In our case, since the two reduced density matrices on the commitment register are classical, so the fidelity between them is
	\begin{equation}
		\label{eq:binding-fidelity}
		\paren{\sum_y\sqrt{\frac{\Pr[H(x) = y]}{2^{6\lambda}}}}^2 \le 2^{-\lambda}
	\end{equation}
	by applying Cauchy--Schwarz on $y$'s in the image of $H$.
	The rest follows the same proof as \cite[Proposition 4.4]{CKR16-commitments}, thus we get that our commitment after taking $\lambda$-fold parallel repetition is statistical sum binding.

	For hiding, the two reduced density matrices given to the hiding adversary exactly corresponds to the security game of the pseudorandom ensemble.
	Finally taking a parallel repetition also preserves hiding by a standard hybrid argument.
\end{proof}

We remark that this construction can be straightforwardly ``derandomized'' using a (post-quantum) pseudorandom generator, and thus eliminating the inefficiency at the cost of assuming the security of the pseudorandom generator.

In fact, if we are only aiming for a commitment scheme with quantum auxiliary input and do not require the classical description to be computable in exponential time, any non-trivial computationally indistinguishable pair of quantum states suffices.
The construction is essentially the same except that the receiver does a SWAP test with the corresponding pure state for decommitments to either bits.

\paragraph{Classical impossibility.}

For convenience, we focus on non-interactive statistically-binding commitments for the classical case. %
This is a fair comparison since such a commitment with classical auxiliary information does exist assuming existence of one-way functions, by applying averaging argument to the first message of Naor commitment \cite{Naor91-commitment}.
We leave improving this impossibility to future work.
The main insight is that having sample access to the oracle suffices to break the security.

\begin{theorem}
	\label{thm:classical-impossibility}
	If there exists a computationally-hiding statistically-binding non-interactive classical commitment scheme where all parties share access to a sampling oracle, then $\mathsf{NP} \not\subseteq \mathsf{P/poly}$.
\end{theorem}
\begin{proof}
	For simplicity, we consider that committer and receiver each gets a private sample $s, r$ respectively from the sampling oracle; to commit to bit $b$, the committer computes a deterministic function $\Commit(s, b) = \tau$ and sends $\tau$ to the receiver; to decommit, the committer sends $s, b$ and the receiver computes a deterministic function $\Reveal(s, b, \tau, r) \in \{\top, \bot\}$ indicating accept or reject the revealed bit $b$.
	This is without loss of generality since we can always have the oracle give multiple samples per query and pad the oracle with uniform random bits so that the only source of randomness is from the oracle.

	Assume $\mathsf{NP} \subseteq \mathsf{P/poly}$, we show how to efficiently break the hiding.
	Let each sample be $\ell$ bits long.
	The malicious receiver, on input $\tau$ and $(\ell + 1)$ independent samples $\vec r := r_0, ..., r_\ell$ from the sampling oracle, %
		predicts the commitment is to $0$ if there exists $s_0$ such that $\Reveal(s_0, 0, \tau, r_i) = \top$ for all $i$, and predicts to $1$ otherwise.
	This is also an $\mathsf{NP}$ language since $\Reveal$ is efficient and $\ell$ is polynomial.
	Then in the case when the committer commits to $0$, by completeness with overwhelming probability over $\tau$, $\Pr_r[\Reveal(s_0, 0, \tau, r)]$ is negligibly close to 1, so this receiver always predicts $0$ except with negligible probability.
	On the other hand, when the committer commits to $1$, by statistical binding, with overwhelming probability over $\tau$, for every $s_0$ we have that $\Pr_r[R(s_0, 0, \tau, r)]$ is negligible, and by independence of samples, $\Pr_{\vec r}[\forall i : \Reveal(s_0, 0, \tau, r_i)] \le 2^{-\ell - 1}$ for all sufficiently large $\lambda$, thus by union bound, $\Pr_{\vec r}[\exists s_0 \forall i : \Reveal(s_0, 0, \tau, r_i)] \le \frac12$.
	This gives a hiding adversary with advantage at least $\frac12 - \mathsf{negl}$.
	\end{proof}

Not accounting for implementing the $\mathsf{NP}$ oracle, the malicious receiver we construct in this proof uses roughly quadratic space compared to the honest parties.
Curiously, this matches the quadratic space lower bound proven for bit commitments in the bounded storage model \cite[Section 6]{GZ19-bsm} so this suggests that taking multiple samples is likely required for such attacks in general.

\section{Commitment in the unclonable common random state model}
\label{sec:crs}

The common random string (CRS) model is a commonly considered relaxation of the standard trustless model where the only trust in the setup is that a classical string is uniformly sampled and then published.
This model was first proposed in the context of non-interactive zero knowledge (NIZK) \cite{BFM88-nizk} since interesting NIZK is impossible in the trustless model.

In this work, we introduce a quantum analogue of the CRS model that we call the unclonable common random state ($\UCRS$) model.
In the $\UCRS$ model, the only trust in the setup is that a random pure state is drawn from a state distribution, and then many copies of that pure state are made available to all parties.
Furthermore, there should be a way to efficiently generate any polynomial (but a priori unbounded) number of the common state.
We emphasize that unclonability only indicates a lack of the cloning functionality of the common state but the scheme's security does not rely on the common state being unclonable --- in fact in this model, any malicious party is allowed to inefficiently process the classical description of the random state before the protocol begins.

Going back to the commitment construction of \Cref{thm:main}, we note that the magic state is simply a uniform superposition query on a function $H$, and even using a random function $H$ is secure except with inverse exponential probability.
Therefore, that construction is indeed also secure in the $\UCRS$ model, so the only missing piece of the puzzle is to show efficient sampling of the magic state.
This is tricky since the classical description requires an exponential size and there are a doubly exponential number of possible states to sample from.

Nevertheless, we show that there is an efficient way to statefully sample these states.
To do this, we invoke Zhandry's compressed oracle technique \cite{Zha19-record}.

\begin{lemma}[{\cite{Zha19-record}}]
	There exists a stateful simulation oracle $\mathsf{CStO}$ that perfectly simulate any number of (quantum) queries to a random function $H: \{0, 1\}^n \to \{0, 1\}^m$ for any $n, m$.
	Furthermore, the $t$-th query can be processed in time polynomial in $mt$, and the state (also called the database) after $t$ queries consists of $(m + n + 1) \cdot t$ qubits.%
\end{lemma}

Since the simulation is perfect, the commitment constructed in \Cref{thm:main}, after replacing $H$ with a truly random function, still works.
The simulation is also efficient since both $m$ and $t$ are polynomial in $\lambda$.
Thus we arrive at the following corollary.

\begin{corollary}
	\label{cor:ucrs-comm}
	There exists a computationally-hiding statistically-binding non-interactive quantum commitment scheme in the $\UCRS$ model.
\end{corollary}

Interestingly, the hiding of the commitment in fact becomes statistical if at most a polynomial number of states are given out to the adversary, and the proof of this is deferred to \Cref{sec:nonuniform-proof}.
Therefore, we get a completely statistically secure commitment if we in addition trust the setup to not generate too many copies of the magic state.

\begin{proposition}
	\label{prop:stathiding}
	Assuming the receiver has at most $P$ copies of the magic state on $H$, then the commitment scheme of \Cref{cor:ucrs-comm} is $\paren{8\sqrt2 \cdot \sqrt{\frac PN}}$-statistically hiding.
\end{proposition}

Another interesting consequence of this is that for our scheme, we cannot trust either party to distribute the magic state.
If the receiver chooses the magic state for the committer, then computational hiding can be trivially broken by picking a bad magic state like the all zero state.
If the committer chooses the magic state for the receiver, then \Cref{prop:stathiding} shows that this scheme is in fact statistically hiding and thus not statistical binding.
In fact, it is not even computationally binding as we show below.

\begin{proposition}
	\label{prop:notcbinding}
	Take the commitment scheme from \Cref{thm:main} but instead have the committer choose $H$ as a random function.
	Then this commitment scheme has computational sum binding error of at least $1 - O(t/\sqrt N)$ even after taking $t$-fold parallel repetition using the same random function.
\end{proposition}
\begin{proof}[Proof sketch]
	We sketch how to efficiently break sum binding, even if the scheme is repeated $t$ times in parallel.
	Consider a binding adversary that commits to 0 honestly using the compressed oracle.
	Certainly by sending the decommitment registers honestly, the receiver would accept 0 with probability 1.
	We now show how to decommit to 1 with probability $1 - O(t/\sqrt N)$.
	First, we measure the $x$ for every magic state (including the ones held by the decommitter are also measured, which is okay since the receiver does not touch those registers when checking decommitment to 1 so they are essentially traced out) and abort if they are not all distinct, which happens with probability at most $O(t^2/N)$ by collision probability.
	Otherwise, every magic state holds a distinct $x$.
	For each fold, the decommitment register contains some $x_i$: we apply $\mathsf{StdDecomp}_{x_i}$ and search in the database where the entry $x_i$ occurs and send the corresponding image register as decommitment to 1, which is maximally entangled with the corresponding commitment register.
	To make this into a malicious committer that does not measure the $x$'s held by the receiver, we note that this measurement is only used to conditionally abort the committer, so the success probability of the same adversary except that it never aborts is still $1 - O(t/\sqrt N)$ by gentle measurement.
\end{proof}

\section{Eliminating trust with preprocessing}
\label{sec:trustless}

We first give a $\mathsf{PromiseQMA}$ upper bound on the complexity of verifying the computational insecurity of $H$ up to a constant multiplicative loss.
Therefore, we can check if a function $H$ is secure in doubly exponential time, or even exponential time if $\mathsf{BQP} = \mathsf{QMA}$.
Another consequence is that a $\mathsf{\Sigma_2P^{QMA}}$ machine can also find the lexicographically smallest $H$ that is secure.
Thus we can complete the preprocessing phase from the construction of \Cref{thm:main} in doubly-exponential time, or even exponential time if $\mathsf{BQP} = \mathsf{QMA}$.

\begin{proposition}
	\label{prop:qma-security}
	The language $L$, consisting of all functions $H$ such that the commitment constructed in \Cref{thm:main} using $H$ is insecure, is in $\mathsf{PromiseQMA}$.
\end{proposition}
\begin{proof}
	We prove this by constructing a $\mathsf{PromiseQMA}$ verifier.
	The verifier, on input $H$ (of length $N\log M$) and a witness $\ket{C}$ (a distinguisher quantum circuit of length $S$ which is polynomial in $|H|$), samples a random bit $b$ and runs the universal quantum circuit on either $(\ket{C}, H(x))$ for a random $x \in [N]$ or $(\ket{C}, y)$ for a random $y \in [M]$, and accepts if the universal quantum circuit predicts $b$ correctly.
	We set completeness to be $\frac12 + 2^{-\lambda}$ and soundness to be $\frac12 + 2^{-\lambda - 1}$, so the gap is $2^{-\lambda - 1}$ which is inverse polynomial in $|H|$.

	If the output of $H$ is pseudorandom, then we have that $H \not\in L$ as desired.
	On the other hand, if there exists an $S$-sized witness for $H$ that distinguishes with advantage higher than $2 \cdot 2^{-\lambda}$, then $H \in L$.
\end{proof}

We now give an alternative \emph{unconditional} approach to eliminate trust on the magic state in uniform exponential time, but at the cost of introducing exponential communication between the parties in the preprocessing phase.%

\begin{theorem}
	There exists a computationally-hiding statistically-binding non-interactive quantum commitment scheme with an exponential communication preprocessing phase.
\end{theorem}
\begin{proof}
	We again adapt the commitment from \Cref{thm:main} by adding a preprocessing phase to generate the magic states for both parties.
	More specifically, the sender samples a random function $H$ on their own, then send the classical description of $H$ to the receiver.
	Afterwards, they generate multiple copies of the magic state for that function on their own before the protocol begins.
	(It is important that they generate the quantum magic state on their own to not run into the impossibility of \Cref{prop:notcbinding} above.)

	Computational hiding is preserved since in this case $H$ is honestly generated, and the proof of \Cref{thm:main} actually suffices to show that the commitment is statistically binding for any $H$: specifically, the fidelity computed in \eqref{eq:binding-fidelity} is bounded by $2^{-\lambda}$ for any $H$.
\end{proof}

\ifanonymous\else
\section*{Acknowledgements}

The author thanks Ran Canetti, William Kretschmer, Qipeng Liu, Daniel Wichs, as well as Tomoyuki Morimae, Barak Nehoran, and Takashi Yamakawa~\cite{MNY23-qaicomm} for their invaluable feedback on an earlier draft of this work.
Special gratitude is owed to Qipeng Liu for proposing to augment \Cref{thm:main} with compressed oracles (\Cref{cor:ucrs-comm}), and Tomoyuki Morimae, Barak Nehoran, and Takashi Yamakawa~\cite{MNY23-qaicomm} for their helpful discussions contributing to the development of \Cref{sec:simulation}.
The author also thanks Yilei Chen and Peihan Miao for their helpful discussions.
The work is supported by DARPA under Agreement No.\ HR00112020023.
\fi

\ifshort\else
\printbibliography
\fi

\appendix

\ifshort\else

\fi

\section{Simulation security}
\label{sec:simulation}

Simulation security captures the security of a primitive using the real-ideal world paradigm more precisely than the game-based security definitions (which are usually used for hardness assumptions) and is the default security notion in the context of zero knowledge and secure multiparty computations.
In this appendix, we discuss how to further augment our commitment scheme following the template of Bartusek, Coladangelo, Khurana, and Ma~\cite{BCKM21-mpc} so that it satisfies simulation security.

The simulation security for a bit commitment is morally trying to capture the following ideal world: (1) in the commit phase, the committer sends a bit $b$ to the ideal functionality; (2) in the reveal phase, the committer asks the ideal functionality to open and the ideal functionality sends the bit $b$ to the receiver.
More specifically, the security against receiver is called equivocality, which states that the commitment can be simulated in a way that $b$ is only determined at the beginning of the reveal phase.
The security against committer is called extractability, which states that the commitment can be simulated in a way that $b$ can be extracted from the committer after the commit phase completes.

\paragraph{$\varepsilon$-simulation security for inefficient preprocessing.}
We first consider augmenting the base protocol (\Cref{thm:main}) with simulation security.
We note some caveats before proceeding.

The first caveat is that we allow the simulator to take a few copies of the auxiliary input state as additional inputs.
Intuitively, this means that a malicious party can come out of the protocol obtaining a few extra copies of the magic state.
Indeed, if we look at the construction of \Cref{thm:main}, the receiver after an honest interaction gains one copy of the magic state from the committer for each fold of repetition, and there does not seem to be a way for the committer to certifiably retrieve the state back.
We believe that this weakening is still meaningful and non-trivial since (1) the magic state is supposed to be public knowledge anyways (everyone should have many copies), and (2) the simulation security still guarantees that the ``real'' input is hidden from the other party.

The second caveat is that we only achieve $\varepsilon$-simulation security (with quantum auxiliary information $\ket{aux}$), which states that there is an efficient simulator $S$ and some polynomial $t$ such that for every adversary $A$ (that is possibly entangled with the distinguisher) and every $\varepsilon$, the view outputted by $(I \otimes S)(A, \ket{aux}^{\otimes t(1/\varepsilon)})$ is distinguishable from the real view except with advantage no more than $\varepsilon$, where $S$ does not touch the distinguisher's private register.
We stress that the protocol itself is independent of $\varepsilon$.
Furthermore, this still suffices for almost any game-based application that needs simulation security although with a larger polynomial security loss.
To see this, for example, suppose an adversary can break a downstream game-based security with some non-negligible probability $p$, then we can set our overall simulation error to be $p/2$ to reach a contradiction.
On the other hand, to prove a downstream $\varepsilon$-simulation security, we can similarly pick a smaller $\varepsilon$ for each fold and invoke hybrid argument.

The main technical ingredient we need is to implement reflection unitary for an arbitrary initial state, which in this case could contain some inefficient auxiliary information.
This is proven in the following lemma, which constructs such an algorithm by using a generalized SWAP test.

\begin{lemma}[Approximate state reflection]
	\label{lem:approxreflect}
	For any pure state $\ket\psi$, let $\mathcal R_\psi$ be the unitary channel for unitary $R_\psi := I - 2 \ketbra\psi$.
	Then there is a uniformly efficient channel $\tilde{\mathcal R}$ such that $\tilde{\mathcal R}(\cdot, \ketbra\psi^{\otimes n})$ is $\sqrt[4]{\frac{64}{n + 1}}$-close to $\mathcal R_\psi$ in diamond norm for all $\psi, n \ge 0$.
\end{lemma}
\begin{proof}
	We describe the algorithm as follows.
	We denote the registers as $\mathsf X_0, ..., \mathsf X_n$ with $\mathsf X_0$ being the input register and the rest being initialized to $\ket\psi$.
	\begin{enumerate}
		\item Initialize a uniform superposition $\ket{+}_{\mathsf N} := \frac1{\sqrt{n + 1}} \sum_{i = 0}^{n} \ket i_{\mathsf N}$.
		\item Controlled on $\mathsf N$ being $\ket i$, swap $\mathsf X_0$ and $\mathsf X_i$.
		\item Controlled on $\mathsf N$ being $\ket+$, apply phase $-1$.
		\item Uncompute step 2.
		\item Trace out everything except $\mathsf X_0$.
	\end{enumerate}
	We begin analyzing the algorithm by considering pure state inputs.
	If the input is $\ket\psi$ then phase $-1$ is correctly applied since steps 2 and 4 do not affect the state.
	If the input is some orthogonal state $\ket\phi$, let $\ket{\phi_i}$ denote the state where $\mathsf X_i$ is $\ket\phi$ and everywhere else is $\ket\psi$.
	Then after step 3, we get the state
	\[ \frac1{\sqrt{n + 1}}\sum_{i = 0}^n \ket{\phi_i}(\ket i - 2\ket+) = \frac1{\sqrt{n + 1}}\sum_i \ket{\phi_i}\ket i - \frac2{n + 1}\sum_{i, j} \ket{\phi_i}\ket j. \]
	Therefore after step 4, we get
	\[ \ket{\tilde\phi} := \paren{1 - \frac2{\sqrt{n + 1}}} \ket{\phi_0}\ket+ - \frac2{n + 1}\sum_{i \neq j}\ket{\phi_i}\ket j. \]
	We compare this state with the expected output and get that the overlap
	\[ (\bra{\phi_0}\bra+)\ket{\tilde\phi} = 1 - \frac2{\sqrt{n + 1}}\paren{1 + \frac{n}{n + 1}} \ge 1 - \frac4{\sqrt{n + 1}}. \]
	By decomposing a general pure state $\ket x = \sqrt p \ket{\psi} + \sqrt{1 - p} \ket{\phi}$ and let $\ket{\psi_n} := \ket\psi^{\otimes n} \ket+$, we have that
	\begin{equation}
		\label{eq:approxreflection-ip}
		(\bra x R_\psi^\dagger \otimes \bra{\psi_n})\ket{\tilde x} \ge 1 - \frac4{\sqrt{n + 1}}
	\end{equation}
	as well by combining the two cases above.
	Furthermore, \eqref{eq:approxreflection-ip} generalizes to a larger entangled pure state by simply applying \eqref{eq:approxreflection-ip} linearly to the Schmidt decomposition.
	For diamond norm, it suffices to consider any input state where the overall entangled state is pure (since we can without loss of generality purify the state for the distinguisher), thus we have that the trace distance between $(I \otimes R_\psi) \ket x \otimes \ket{\psi_n}$ and $\ket{\tilde x}$ is at most
	\[ \sqrt{1 - \paren{1 - \frac4{\sqrt{n + 1}}}^2} \le \sqrt[4]{\frac{64}{n + 1}}. \]
	This completes the proof since trace distance cannot increase after the operation of tracing out the auxiliary registers holding $\ket{\psi_n}$, which is CPTP.
\end{proof}

\begin{theorem}
	$\varepsilon$-simulation secure commitment schemes with quantum auxiliary input exist.
	Furthermore, it can be built from any non-interactive extractable-binding computationally-hiding commitment scheme with quantum auxiliary input.
\end{theorem}
\begin{proof}
	This follows the same construction and proof strategy as \cite{BCKM21-mpc,AQY22-prs} except for one change.
	In particular, Watrous rewinding \cite{Wat06-ZK} was used for a total of $\lambda$ times in the equivocal simulator \cite[Section 4.1]{BCKM21-mpc}, which involved applying a unitary $I - 2\ketbra0$ on a certain private register of the simulator's.
	The purpose of this unitary was to check whether this register returned to all zero state.
	In our context, this register would be initialized to $\ket0 \otimes \ket{aux}^{\otimes t}$ instead for a suitably large zero register and some polynomial $t(\lambda)$, and similarly we need to reflect around this state in order for the analysis to go through\footnote{One na\"ive idea to fix this is to ask the simulator to only reflect the zero part and ignore the auxiliary information. However, this does not work since we can calculate and see that such a rewinding algorithm (without any further changes) would not work for an adversary that picks its challenge by measuring the auxiliary information it receives from the simulator.}.

	We establish $\varepsilon$-equivocality as follows.
	We first consider a simulator that runs the \cite{BCKM21-mpc} equivocal simulator with access to an inefficient reflection oracle: this gives a negligible simulation error, and thus it is at most $\varepsilon/2$ for all sufficiently large $\lambda$.
	We now instantiate this inefficient oracle with \Cref{lem:approxreflect} where $n = \left\lceil 1024(\lambda/\varepsilon)^4 \right\rceil - 1$ and $\ket\psi = \ket0 \otimes \ket{aux}^{\otimes t}$, then by a standard hybrid argument we arrive that the overall simulation error is at most $\varepsilon$.
	Furthermore, the number of copies of $\ket{aux}$ used is $O(t(\lambda/\varepsilon)^4)$ which is polynomial.

	The rest of the proof follows as \cite{BCKM21-mpc} using a similar trick of running the inner simulator with a polynomially smaller error parameter.
\end{proof}

\paragraph{Simulation security in the unclonable common random state model.}
We conclude by remarking on the simulation security of the commitment scheme with trusted setup from \Cref{cor:ucrs-comm}.
First of all, we can still apply the \cite{BCKM21-mpc} transformation, but in this case we would get negligible simulation security since with access to the compressed database register, the simulator can efficiently test/uncompute the magic state.

However, we note that the argument of \Cref{prop:notcbinding} gives an adversary that breaks honest binding when the adversary could control the sampling of the random function using compressed oracles, and thus an analogous argument can show that the commitment of \Cref{cor:ucrs-comm} is in fact already negligibly equivocal even without any further modification to the scheme.

Using similar ideas this commitment is probably also negligibly extractable as well, however, the argument is more complicated since in this case we need to be able to extract any malicious committer (this is unlike the equivocal case where the only freedom a passive but malicious receiver has is to distinguish the views, so it suffices to simply show simulation correctness).
We leave formalizing this to future work.

\section{Post-quantum pseudorandomness}
\label{sec:nonuniform-proof}

We view a quantum query-less circuit with auxiliary input of total size $S$ as an $S$-qubit input fed to a universal quantum circuit, which itself is independent of the random function $H$.
(Indeed we can without loss of generality even take $S$ to be the number of qubits that actually depend on the function $H$.)

The first polynomial upper bound for this problem was established by Chung, Guo, Liu, and Qian~\cite{CGLQ20-tradeoffs} and was subsequently improved by Liu~\cite{Liu23-advice}.
We follow the second work in this proof.
We first recall a game $G$ in a $P$-BF-QROM \cite{Liu23-advice} to be the following:
\begin{enumerate}
	\item[0.] A random function $H: [N] \to [M]$ is sampled uniformly at random.
	\item The adversary starts by making $P$ (quantum) queries to $H$, and then we postselect on measuring its first qubit and obtaining 1 (abort if it is not possible).
		This postselection may affect its residual state as well as the conditional distribution of the random function.
	\item The challenger then samples a random classical challenge to the adversary, using $T_{samp}$ queries.
	\item The adversary produces a response, using $T$ queries.
	\item The challenger outputs a bit indicating accept or reject, using $T_{verify}$ queries.
\end{enumerate}
We say $G$ is $\nu(P, T)$-secure in the $P$-BF-QROM if any adversary with $T$ queries cannot make the challenger accept with probability higher than $\nu$.
In our case, the security game of a pseudorandom ensemble (or PRG) against a query-less adversary corresponds to $T_{samp} = 1$ and $T = T_{verify} = 0$: the challenger flips a random bit and either sends a pseudorandom $H(x)$ (using a single query) or a random $y$, and asks the adversary to predict the bit.

Similarly, an $(S, T)$ non-uniform quantum adversary plays the same security game, except that in step 1 it can do an arbitrary amount of queries but is not allowed to do post-selection, and its output (to be used later in step 3) is restricted to at most $S$ qubits.

\begin{lemma}[{\cite[Lemma 4]{Liu23-advice}, with the coefficient from \cite[Proof of Lemma 5.13]{CGLQ20-arxiv}}]
	\label{lemma:prg-bf}
	The PRG game has $\nu(P, T) = \frac12 + 4\sqrt2 \cdot \sqrt{\frac{P + T^2}N}$ in the $P$-BF-QROM.
\end{lemma}

\begin{theorem}[{\cite[Theorem 5]{Liu23-advice}}]
	\label{thm:bf2qai}
	Any game $G$ that has security $\nu$ in the $P$-BF-QROM has security
	\[ \delta(S, T) \le \min_{\gamma > 0}\mbrace{\nu(P/\gamma, T) + \gamma} \]
	against $(S, T)$ non-uniform adversaries in QROM, where $P = S(T + T_{verify} + T_{samp})$.
\end{theorem}

\begin{proof}[Proof of \Cref{thm:nonuniform}]
	Combining \Cref{lemma:prg-bf} and \Cref{thm:bf2qai}, we find that for any non-uniform algorithm $A$ of size $S$ (that potentially depends on $H$),
	\begin{align*}
		\abs{\E_H\mbracket{\E_x[A(H(x))] - \E_y[A(y)]}}
			&= 2\abs{\delta(S \cdot 1, 0) - \frac12} \\
			&\le 2\min_{\gamma > 0}\mbrace{\sqrt{2^5 \cdot \frac{S}{\gamma N}} + \gamma} \\
			&= 12 \cdot \sqrt[3]{\frac{S}{N}},
	\end{align*}
	showing the bound above.
\end{proof}

\begin{proof}[Proof of \Cref{prop:stathiding}]
	Since each copy of the magic state can be efficiently prepared through a single quantum query to $H$, any distinguishing adversary is a valid query-less adversary in the $P$-BF-QROM, and thus we arrive at the proposition by invoking \Cref{lemma:prg-bf}.
\end{proof}

\ifshort
\printbibliography
\fi

\end{document}